\documentclass[a4paper]{article}

\usepackage{url} 
\usepackage{times} 
\usepackage{verbatim}
\usepackage[T1]{fontenc}
\usepackage[swedish,english]{babel}
\usepackage{amsmath}
\usepackage{amssymb}
\usepackage[dvipdf]{graphicx}
\usepackage{color}
\usepackage{url}

\usepackage{epsfig}
\usepackage{float}
\usepackage[small]{caption}
\usepackage{subfig}
\usepackage{eurosym}
\graphicspath{{pics_p/}}

\def\R{\mathbb{R}}

\def\half{\frac{1}{2}}
\def\F{\mathcal F}
\def\A{\mathcal A}
\def\L{\mathcal L}
\newcommand{\eqdist}{\,{\buildrel d \over =}\,}

\newtheorem{theorem}{Theorem}
\newtheorem{proposition}[theorem]{Proposition}

\newenvironment{proof}
{\begin{trivlist}\item[\, 
{\bf Proof.}]}{{\hfill $\square$}\end{trivlist}}

\title{Risk aggregation and stochastic claims reserving\\ in disability insurance}

\author{
Boualem Djehiche\footnote{Department of Mathematics, KTH Royal Institute of Technology, Sweden, {boualem@kth.se}. }
\and Björn Löfdahl \footnote{Department of Mathematics, KTH Royal Institute of Technology, Sweden, {bjornlg@kth.se}. }
}

\date{\today}

\begin{document}

\maketitle

\begin{abstract}
We consider a large, homogeneous portfolio of life or disability annuity policies. The policies are assumed to be independent conditional on an external stochastic process representing the economic-demographic environment. Using a conditional law of large numbers, we establish the connection between claims reserving and risk aggregation for large portfolios. Further, we derive a partial differential equation for moments of present values. Moreover, we show how statistical multi-factor intensity models can be approximated by one-factor models, which allows for solving the PDEs very efficiently. Finally, we give a numerical example where moments of present values of disability annuities are computed using finite-difference methods and Monte Carlo simulations.

\end{abstract}

\vspace{4mm}

\section{Introduction}
The upcoming Solvency II regulatory framework brings many new challenges to the insurance industry. In particular, the new regulations suggest a new mindset regarding the valuation and risk management of insurance products. Historically, premiums and reserves are calculated under the assumption that the underlying transition intensities of death, disability onset, recovery and so on are deterministic. While the estimations should be prudent, this still implies that the systematic risk, i.e. the risk arising from uncertainty of the future development of the hazard rates, is not taken into account. This may have an impact on pricing as well as capital charges. In the Solvency II standard model, capital charges are computed using a scenario based approach, and the capital charge is given as the difference between the present value under best estimate assumptions, and the present value in a certain shock scenario. As an alternative, insurers may adopt an internal model, which should be based on a Value-at-Risk approach. 

Facing these challenges, a plethora of stochastic intensity models have appeared, in particular for modelling mortality. However, these works have largely focused on either calibration or on pricing a single policy under a suitable market implied measure. The risk management aspect has been left largely untouched, although there are notable exceptions. Dahl \cite{Dahl04} derives a pricing PDE for a wide class of life derivatives under a one-factor stochastic intensity model. Dahl points out that shocks from a one-factor model affect all cohorts equally, and that a multi-factor model across cohorts might be more realistic, although it would not offer any further insights. Dahl and M\o ller \cite{DahlMoller} consider pricing and hedging of life insurance liabilities with systematic mortality risk. Biffis \cite{Biffis05} considers annuities pricing under affine mortality models. Ludkovski and Young \cite{LudkovskiYoung} consider indifference pricing under stochastic hazard. Norberg \cite{Norberg95} derives an ODE for moments of present values assuming deterministic hazard rates. 

While stochastic mortality models have been thoroughly studied in the literature, stochastic disability models have not received the same attention. Levantesi and Menzietti \cite{LevantesiMenzietti} consider stochastic disability and mortality in the Solvency II context. The approach covers both systematic and idiosyncratic risk, and is suitable for small portfolios. Christiansen {\it et al}. \cite{Christiansen} suggest an internal model for Solvency II based on the forecasting technique of Hyndman and Ullah \cite{HyndmanUllah}. The approach includes fitting an intensity model over a range of time periods, and fitting a time series model to the time series of parameter estimates. The future development of the intensities is obtained by forecasting or simulation of the time series model. 

In this paper, we consider a large, homogeneous portfolio of life or disability annuity policies. The policies are assumed to be independent conditional on an external stochastic process representing the economic-demographic environment. Using a conditional law of large numbers, we show that the aggregated annuity cash flows can be approximated by its conditional expectation, an expression much akin to the actuarial reserve formula. This result highlights the connection between risk aggregation and claims reserving for large portfolios. Further, we derive a partial differential equation for moments of these present values. Moreover, we consider statistical multi-factor intensity models, and suggest methods for reducing their dimensionality. Using the so-called mimicking technique introduced by Krylov \cite{Krylov}, we suggest approximating multi-factor models by one-factor models, which allows for solving the PDEs very efficiently. Finally, we give a numerical example where moments of present values of disability annuities are computed using finite-difference methods and Monte Carlo simulations.

The paper is organized as follows. In Section \ref{sec:reserving}, consider an annuity policy under a simple stochastic intensity model. We derive a PDE for computing moments of the random present value of such policies. In Section \ref{sec:risk_agg}, we examine the aggregated cash flows from a large, homogeneous portfolio of insurance policies, and highlight the connection between risk aggregation and claims reserving. In Section \ref{sec:example}, we consider the specific application of disability insurance, and show how a class of statistical models can be incorporated into the pricing PDEs. In Section \ref{sec:numerics}, we present numerical results based on disability data from the Swedish insurance company Folksam.
\section{Stochastic claims reserving}\label{sec:reserving}
Let $\tau^1, \tau^2,\ldots$ be random event times (e.g. times of death or recovery from disability), and let
\begin{equation}
N^k_t = I\{\tau^k \leq t\},\hspace{10pt} k\geq 1.
\end{equation}
\noindent
Further, define the processes 
\begin{equation}
N^k = (N^k_t)_{t\geq 0},\hspace{10pt} k\geq 1,
\end{equation}
\noindent
and let
\begin{equation}
\F^N=(\F_t^N)_{t\geq0}=(\F_t^{N^1}\vee\F_t^{N^2}\vee\ldots)_{t\geq0}
\end{equation}
\noindent
denote the filtration generated by $N^1, N^2, \ldots$. Now, let $Z$ be a stochastic process with natural filtration $\F^Z=(\F_t^Z)_{t\geq0}$. Here, $N^k_t$ denotes the state of an insured individual at time $t$, $\tau^k$ represents the corresponding death or recovery time, and $Z_t$ represents the state of the economic-demographic environment. We assume that $N^1,N^2,\ldots$ are independent conditional on $\F_\infty^Z$, and that the $\F^Z\vee\F^N$-intensity of $N^k$ is the process $\lambda^k$ of the form
\begin{equation}
\lambda^k_t = q(t,Z_t)(1-N^k_{t}), \hspace{10pt} t\geq0.
\end{equation}
\noindent
Consider an annuity policy paying $g(t,Z_t)$ continuously as long as $N^k_t=0$, until a fixed future time $T$. This type of annuity allows for payments from the contract to depend on time as well as the state of the economic-demographic environment. For example, the contract could be inflation-linked and contain a deferred period. The random present value $L^k_t$ of this policy can be written as
\begin{equation}
L^k_t = \int_t^{T}g(s,Z_s)(1-N^k_s)e^{-\int_t^sr(u)du}ds,
\end{equation}
\noindent
where the short rate $r$ is assumed to be adapted to $\F^Z$. Further, the time $t$ reserve for this contract is given by $E[L^k_t|\F^Z_t\vee\F^N_t]$, that is, the expected value given the history of the policy and of the environment. Our goal is to find an efficient way to compute this reserve. First, we need the following result, which is given in a slightly different form in Norberg's concise introduction to stochastic intensity models \cite{Norberg}.
\begin{proposition}\label{prop:conditional_X}
Assume that $E[|\lambda^k_t|] < \infty$  for each $k$, $t\geq0$. Then, for $s\geq t$,
\begin{equation}\label{eq:conditional_X}
E[1-N^k_s|\F^Z_s\vee\F^N_t] = P(\tau^k>s|\F^Z_s\vee\F^N_t) = (1-N^k_t)e^{-\int_t^sq(u,Z_u)du}.
\end{equation}
\noindent

\end{proposition}
\begin{proof}
First, note that the process $(M^k_s)_{s\geq 0}$ defined by
\begin{equation}\label{eq:N_mg}
M^k_s = N^k_s-\int_0^s\lambda^k_udu
\end{equation}
\noindent
is a $\F^Z\vee\F^N$-martingale \cite[p. 106]{Norberg}. For $s\geq t$, let $Y^k_s = P(\tau^k>s|\F^Z_s\vee\F^N_t) = E[1-N^k_s|\F^Z_s\vee\F^N_t]$. Using \eqref{eq:N_mg}, we have
\begin{align}
Y^k_s &= E[1-N^k_s+\int_0^s\lambda^k_udu-\int_0^s\lambda^k_udu|\F^Z_s\vee\F^N_t]\nonumber\\
&= 1- N^k_t + \int_0^t\lambda^k_udu -E[\int_0^s\lambda^k_udu|\F^Z_s\vee\F^N_t]\nonumber\\
&= 1- N^k_t + \int_0^t\lambda^k_udu -\int_0^t\lambda^k_udu -E[\int_t^s\lambda^k_udu|\F^Z_s\vee\F^N_t]\nonumber\\
&= 1- N^k_t - \int_t^sq(u,Z_u)E[1-N^k_u|\F^Z_s\vee\F^N_t]du\nonumber\\
&= 1- N^k_t - \int_t^sq(u,Z_u)E[1-N^k_u|\F^Z_u\vee\F^N_t]du\nonumber\\
&= 1- N^k_t - \int_t^sq(u,Z_u)Y^k_udu.
\end{align}
\noindent
Differentiating the above expression, we obtain
\begin{align}
\begin{cases}
dY^k_s = -q(s,Z_s)Y^k_sds,\hspace{10pt} s>t,\\
Y^k_t = 1-N^k_t,
\end{cases}
\end{align}
\noindent
with solution $Y^k_s = (1-N^k_t)e^{-\int_t^sq(u,Z_u)du}$.
\end{proof}
\noindent
Using Proposition \ref{prop:conditional_X}, we immediately obtain
\begin{align}
E[L^k_t|\F^Z_t\vee \F^N_t]&= E[E[L^k_t|\F^Z_T\vee \F^N_t]|\F^Z_t\vee \F^N_t]\nonumber\\
&= E[E[\int_t^{T}g(s,Z_s)(1-N^k_s)e^{-\int_t^sr(u)du}ds|\F^Z_T\vee \F^N_t]|\F^Z_t\vee \F^N_t]\nonumber\\
&= (1-N^k_t)E[\int_t^{T}g(s,Z_s)e^{-\int_t^sq(u,Z_u)du}e^{-\int_t^sr(u)du}ds|\F^Z_t\vee \F^N_t].
\end{align}
\noindent
Note that if the environment process $Z$ is replaced by a deterministic function, the functional $V_t$ defined by 
\begin{equation}\label{eq:lln_cashflows}
V_t = \int_t^{T}g(s,Z_s)e^{-\int_t^sq(u,Z_u)du}e^{-\int_t^sr(u)du}ds
\end{equation}
\noindent
corresponds to the time $t$-reserve of a policy paying $g$ monetary units continuously. Now, since $q$ and $g$ are functions of the stochastic process $Z$, $V_t$ is a random variable, and the reserve depends on the distribution of $V_t$. In the case where $Z$ is a Markov process, a natural candidate for the time $t$ reserve of an active contract is the function $v(t,z)$ given by
\begin{eqnarray}\label{eq:reserve}
v(t,z) = E[V_t|Z_t = z] = E^{t,z}\Big[\int_t^T g(s,Z_s)e^{-\int_t^sq(u,Z_u)du}e^{-\int_t^sr(u)du}ds\Big].
\end{eqnarray}
\noindent
Let $\bar q(t,z) = q(t,z)+r(t)$, and assume that $\bar q$ is lower bounded, $g$ is continuous and bounded, and that $Z$ is a Markov process with infinitesimal generator $\A$. Then, $v(t,z)$ given by \eqref{eq:reserve} is a Feynman-Kac functional, satisfying the backward PDE
\begin{align}\label{eq:reserve_pde}
\begin{cases}
-\frac{\partial v}{\partial s} +\bar q(s,z)v = \A v + g(s,z), \hspace{10pt} t\leq s<T, \\
v(T,z) = 0.
\end{cases}
\end{align}
\noindent
For risk management purposes, it is not enough to be able to compute expected values. Often, it is necessary to estimate moments or quantiles. Moments of $V_t$ can be found using the following result.

\begin{proposition}\label{prop:moments_pde}
Let $\bar q = q+r$, and assume that $\bar q$ is lower bounded, $g$ is continuous and bounded, and that $Z$ is a Markov process with generator $\A$. Then, for $n\geq 1$, $v_n(t,z) = E^{t,z}[V_t^n]$ satisfies the backward PDE
\begin{align}\label{eq:moments_pde}
\begin{cases}
-\frac{\partial v_n}{\partial s} +n\bar q(s,z)v_n = \A v_n + ng(s,z)v_{n-1}, \hspace{10pt} t\leq s<T, \\
v_n(T,z) = 0,
\end{cases}
\end{align}
\noindent
where, naturally, $v_0(t,z) = E^{t,z}[V_t^0] = 1.$
\end{proposition}

\begin{proof}
Differentiating $V_t$, we obtain
\begin{align}\label{eq:v_ito}
& dV_t = (\bar q_tV_t-g_t)dt.
\end{align}
\noindent
Therefore,
\begin{align}\label{eq:moments_ito2}
& d(V^n_t) = nV_t^{n-1}(\bar q_tV_t-g_t)dt = (n\bar q_tV_t^n - ng_tV_t^{n-1})dt.
\end{align}
\noindent
Multiplying with the integrating factor $e^{-\int_0^tn\bar q_udu}$, integrating and using $V^n_T$ = 0, we have
\begin{align}\label{eq:moments_ito3}
V^n_t =  \int_t^Tng(s,Z_s)V_s^{n-1}e^{-\int_t^sn\bar q_udu}ds.
\end{align}
\noindent
Taking conditional expectations and using the Markov property of $Z$,
\begin{align}\label{eq:moments_ito4}
E^{t,z}[V^n_t] &=  E^{t,z}[\int_t^Tng(s,Z_s)V_s^{n-1}e^{-\int_t^sn\bar q_udu}ds]\nonumber\\
&= E^{t,z}[\int_t^TE[ng(s,Z_s)V_s^{n-1}e^{-\int_t^sn\bar q_udu}|\F^Z_s]ds]\nonumber\\
&=E^{t,z}[\int_t^Tng(s,Z_s)E[V_s^{n-1}|Z_s]e^{-\int_t^sn\bar q_udu}ds]\nonumber\\
&=E^{t,z}[\int_t^Tng(s,Z_s)v_{n-1}(s,Z_s)e^{-\int_t^sn\bar q_udu}ds].
\end{align}
\noindent
From the Feynman-Kac formula, it follows immediately that $v_n(t,z)$ satisfies the PDE \eqref{eq:moments_pde}, see e.g. Friedman \cite[Theorem 5.3]{friedman2012stochastic} for details.
\end{proof}
\noindent
Proposition \ref{prop:moments_pde} can be used to find the $k$'th moment of $V_t$ by solving the PDE \eqref{eq:moments_pde} for $n=1,\ldots,k$ iteratively. This is useful since it is often faster to numerically solve a PDE than to perform a Monte Carlo simulation, especially for this type of path-dependent problem.

\section{Risk aggregation}\label{sec:risk_agg}
We now consider the risk aggregation problem. For a portfolio consisting of annuity policies for the population $N^1,N^2,\ldots,N^n$, the random present value $L^{(n)}_t$ becomes
\begin{equation}\label{eq:L_t}
L^{(n)}_t=\sum_{k=1}^nL_t^k = \sum_{k=1}^n\int_t^{T}g(s,Z_s)(1-N_s^k)e^{-\int_t^sr(u)du}ds.
\end{equation}
\noindent
We will now investigate the properties of $L^{(n)}$ as the number of policies grows large.
\begin{proposition}\label{prop:L_t convergence}
Conditional on $\F^Z_T\vee\F^N_t$,
\begin{equation}
\lim_{n\rightarrow \infty}\frac{1}{n}L^{(n)}_t - \frac{1}{n}\sum_{k=1}^n(1-N^k_t)V_t = 0 \hspace{10pt} a.s.,
\end{equation}
\noindent
where $V_t$ is given by \eqref{eq:lln_cashflows}.
\end{proposition}

\begin{proof}
Since $N^1_s,N^2_s,\ldots$ are independent conditional on $\F^Z_s\vee\F^N_t$ with
\begin{align}
\sum_{k=1}^\infty\frac{E[(N^k_s-E[N^k_s|\F^Z_s\vee\F^N_t])^2|\F^Z_s\vee\F^N_t]}{k^2} \leq \sum_{k=1}^\infty\frac{1}{k^2} < \infty,
\end{align}
\noindent
it follows from the conditional Law of Large Numbers (see Prakasa Rao \cite[Theorem 6]{PrakasaRao}) that, conditional on $\F^Z_s\vee\F^N_t$,
\begin{equation}\label{eq:sum_N_limit}
\lim_{n\rightarrow \infty}\frac{1}{n}\sum_{k=1}^nN^k_s - E[\frac{1}{n}\sum_{k=1}^nN^k_s|\F^Z_s\vee\F^N_t] = 0 \hspace{10pt} a.s.
\end{equation}
\noindent
This implies that, conditional on $\F^Z_T\vee\F^N_t$,
\begin{align}\label{eq:L_t convergence}
& \frac{1}{n}L^{(n)}_t - E[\frac{1}{n}L^{(n)}_t|\F^Z_T\vee\F^N_t] \nonumber\\
=& \frac{1}{n}\sum_{k=1}^n\int_t^{T}g(s,Z_s)(1-N_s^k)e^{-\int_t^sr(u)du}ds\nonumber\\
-&  E[\frac{1}{n}\sum_{k=1}^n\int_t^{T}g(s,Z_s)(1-N_s^k)e^{-\int_t^sr(u)du}ds|\F^Z_T\vee\F^N_t]\nonumber\\
=& \frac{1}{n}\sum_{k=1}^n\int_t^{T}g(s,Z_s)(E[N_s^k|\F^Z_s\vee\F^N_t]-N_s^k)e^{-\int_t^sr(u)du}ds \rightarrow 0 \hspace{10pt} a.s.,
\end{align}
\noindent
by \eqref{eq:sum_N_limit} and the conditional dominated convergence theorem. Now, using Proposition \ref{prop:conditional_X}, we have
\begin{align}\label{eq:conditional_L}
E[\frac{1}{n}L^{(n)}_t|\F^Z_T\vee \F^N_t] &= \int_t^{T}\frac{1}{n}\sum_{k=1}^nE[g(s,Z_s)(1-N^k_s)|\F^Z_s\vee \F^N_t]e^{-\int_t^sr(u)du}ds\nonumber\\
&=\int_t^{T}\frac{1}{n}\sum_{k=1}^n(1-N^k_t)g(s,Z_s)e^{-\int_t^sq(u,Z_u)du}e^{-\int_t^sr(u)du}ds\nonumber\\
&=\frac{1}{n}\sum_{k=1}^n(1-N^k_t)\int_t^{T}g(s,Z_s)e^{-\int_t^sq(u,Z_u)du}e^{-\int_t^sr(u)du}ds\nonumber\\
&=\frac{1}{n}\sum_{k=1}^n(1-N^k_t)V_t.
\end{align}
\noindent
The claim follows from \eqref{eq:L_t convergence} and \eqref{eq:conditional_L}.
\end{proof}
\noindent
When the portfolio is large enough, Proposition \ref{prop:L_t convergence} motivates the approximation
\begin{equation}
L^{(n)}_t \approx \sum_{k=1}^n(1-N^k_t)V_t.
\end{equation}
\noindent
Hence, in order to determine the distribution of the present value of the portfolio given the history of the environment and the policies, it suffices to consider the random variable $V_t$.
Indeed, all the individual risks are diversified away, and only the systematic risk, that is, the risk that the economic-demographic environment changes, remains. This is formalized through the random variable $V_t$. In particular, an approximate $p$-quantile of the random present value of the portfolio is given by the relation
\begin{equation}\label{eq:quantile}
F^{-1}_{L^{(n)}_t}(p) \approx F^{-1}_{\sum_{k=1}^n(1-N^k_t)V_t}(p) = \sum_{k=1}^n(1-N^k_t)F^{-1}_{V_t}(p),
\end{equation}
\noindent
where the equality follows from the positive homogeneity of the quantile function. This result is analogous to the loan portfolio risk result of Vasicek \cite{VasicekRisk}, which is the foundation of the Basel regulatory credit risk framework. In the Basel framework, the homogeneity requirement of the portfolio is relaxed to allow for efficient approximation of portfolio Value-at-Risk and capital allocation, which possibly suggests that it can also be considered in this application.

Properties of $V_t$ can be investigated using simulation or PDE techniques. Further, the time $t$ reserve for the entire portfolio is given by
\begin{equation}
E[L^{(n)}_t|\F^Z_t\vee \F^N_t]= E[E[L^{(n)}_t|\F^Z_T\vee \F^N_t]|\F^Z_t\vee \F^N_t] = \sum_{k=1}^n(1-N^k_t)E[V_t|\F^Z_t],
\end{equation}
\noindent
and the amount of money allocated to each active policy at time $t$ is simply $E[V_t|\F^Z_t]$. Based on these considerations, the problem of risk aggregation is closely connected to the problem of claims reserving. 

We conclude this section with some comments regarding the Solvency II framework. In the Solvency II standard model, capital charges are computed using a scenario based approach, and the capital charge is given as the difference between the present value under best estimate assumptions, and the present value in a certain shock scenario. As an alternative, insurers may adopt an internal model, which should be based on a Value-at-Risk approach over a one-year time horizon. For instance, the capital charge may be taken to be the Economic Capital, i.e. the difference between the time $t$ value and the $p$-quantile of the value at time $t+1$. We stress the fact that the approximate portfolio quantile given by \eqref{eq:quantile} represents the risk over the entire policy period, i.e. it can be used to compute Value-at-Risk over $T-t$ years. Thus, a topic for future research would be to find an extension of the above result, compatible with the Solvency II framework.

\section{Application to disability insurance}\label{sec:example}
In this section, we consider an example from disability insurance. We seek to compute moments of $V_t$ for which the process $Z$, representing the economic-demographic environment, is constructed from a generalized linear model for disability recovery probabilities. For simplicity, we will assume that the short rate is deterministic. As we will see below, $Z$ is typically non-Markov, and we cannot directly use the Feynman-Kac formula to compute moments of $V_t$. We will consider two possible solutions to this problem. First, we construct a multivariate Markov process with $Z$ as one of its component. This turns out to work well in some special cases. Second, we will rely on the so-called mimicking technique to obtain a reliable approximation of $V_t$.

\subsection{A stochastic termination model}\label{sec:stochastic_termination}
Following Aro, Djehiche and L\"ofdahl \cite{ADL13}, the probability $p_{\nu_t}(x,d)$ that the disability of an individual with disability inception age $x$ and disability duration $d$ is terminated within $[d,d+\Delta d)$ is given by
\begin{equation}\label{eq:p_nu}
p_{\nu_t}(x,d) = \frac{\exp\big\{\sum_{i=1}^n \phi^i(x)\sum_{j=1}^m\psi^j(d)\nu^{i,j}_t\big\}}{1+\exp\big\{\sum_{i=1}^n \phi^i(x)\sum_{j=1}^m\psi^j(d)\nu^{i,j}_t\big\}},
\end{equation}
\noindent
where $\phi$ and $\psi$ are basis functions in $x$ and $d$, respectively, and $\nu$ is an $n\times m$-dimensional stochastic process. For simplicity, the termination intensity $q(d,\nu_t)$ is approximated to be piecewise constant over a small time period $\Delta d$, i.e. it is given by the relation
\begin{equation}\label{eq:probability_intensity}
p_{\nu_t}(x,d)=1-\exp\big\{-q(d,\nu_t)\Delta d\big\}.
\end{equation}
\noindent
In the present context, the duration $d$ is simply assumed to be 0 at time $t=0$. Using this, together with \eqref{eq:p_nu}-\eqref{eq:probability_intensity}, we obtain, for a fixed $x$ and $\Delta d$, the following approximation for the intensity $q$:
\begin{equation}
q(t,\nu_t) = \frac{1}{\Delta d}\log\Big(1+\exp\big\{\sum_{i=1}^n \phi^i(x)\sum_{j=1}^m\psi^j(t)\nu^{i,j}_t\big\}\Big).
\end{equation}
\noindent
Given a suitable stochastic process form for $\nu$, we may solve the PDE \eqref{eq:moments_pde} with $nm$ space dimensions. However, this is not very efficient when $nm$ is large. To obtain a more tractable model, we will try to reduce the number of dimensions. 

\subsection{Reducing the dimensionality}
Define the process $Z = \{Z_t\}_{t\geq0}$ by
\begin{equation}
Z_t = \sum_{i=1}^n \phi^i(x)\sum_{j=1}^m\psi^j(t)\nu^{i,j}_t,
\end{equation}
\noindent
and define the function $f$ by
\begin{equation}\label{eq:f}
f(\cdot) = \frac{1}{\Delta d}\log(1+\exp(\cdot)),
\end{equation}
\noindent
so that we have
\begin{equation}\label{eq:f_and_q}
q(t,\nu_t) = f(Z_t),\hspace{10pt} t\geq0.
\end{equation}
It is easily seen that we can rewrite $Z_t$ on vector form as
\begin{equation}\label{eq:Z}
Z_t = a(t)^T\nu_t,
\end{equation}
\noindent
with 
\begin{align}
a(t)^T &= (\phi^1(x)\psi^1(t),\ldots,\phi^n(x)\psi^m(t)),\\
\nu_t& = (\nu_t^{1,1},\ldots,\nu_t^{n,m}).
\end{align}
\noindent
From now on, we restrict our attention to the case where $\nu$ can be written as
\begin{equation}\label{eq:nu}
\nu_t = \nu_0 + \mu t + AW_t,
\end{equation}
\noindent
where $W$ is an $nm$-dimensional standard Brownian motion with independent components, $\mu\in\R^{nm}$ and $A\in\R^{nm\times nm}$ is the Cholesky factorization of the covariance matrix $\Sigma$ of $\nu$. In principle, any dynamic for $\nu_t$ is possible. The random walk is a natural choice, since it is easy to fit and simulate, and has been the model of choice in e.g. Christiansen {\it et al.} \cite{Christiansen}. If $a$ is locally bounded , this modelling choice guarantees that the assumption in Proposition \ref{prop:conditional_X} is satisfied, since, in view of \eqref{eq:f}-\eqref{eq:f_and_q}, we have
\begin{equation}
E[f(Z_t)] = E[\frac{1}{\Delta d}\log(1+e^{Z_t})] \leq \frac{\log{2}}{\Delta d} + \frac{1}{\Delta d}E[|Z_t|] < \infty.
\end{equation}
\noindent
Next, consider the dynamics of $Z$. The It\^o formula yields, using \eqref{eq:nu} and \eqref{eq:Z},
\begin{equation}
dZ_t = (\dot a^T\nu_t + a^T\mu)dt + a^TAdW_t,
\end{equation}
\noindent
provided that $\dot a$ exists. This expression cannot directly be written on the form
\begin{equation}\label{Z_dynamics}
dZ_t = \alpha(t,Z_t)dt + \gamma(t)dW_t,
\end{equation}
\noindent
and therefore it is not a 1-dimensional It\^o diffusion. In general, it is not even a Markov process. This is due to the time dependence of $a$, a property which originates from the fact that the termination intensity depends on the duration of the illness. This property cannot easily be relaxed.

To remedy this, it may be possible to construct a process $\widehat Z$ of the form \eqref{Z_dynamics}, identical to $Z$ in law. This would imply 
\begin{align}
V_t &= \int_t^{T}g(s,\nu_s)e^{-\int_t^sq(u,\nu_u)du}e^{-\int_t^sr(u)du}ds \\
&\eqdist \int_t^{T}g(s,\widehat Z_s)e^{-\int_t^sf(\widehat Z_u)du}e^{-\int_t^sr(u)du}ds =: \widehat V_t,
\end{align}
\noindent
and, more importantly, that $\L(V) = \L(\widehat V)$, i.e. that the processes $V$ and $\widehat V$ are identical in law. According to \O ksendal \cite[Theorem 8.4.3]{oksendal2003stochastic}, $\L(\widehat Z) =\L(Z)$ if and only if
\begin{align}\label{eq:oksendal_cond}
\alpha(t,Z_t) &= E[\dot a^T\nu_t + a^T\mu| \F^Z_t]\\
\gamma^2(t) &= a^TAA^Ta.
\end{align}
\noindent
Unfortunately, the conditional expectation \eqref{eq:oksendal_cond} is in general not easy to compute. We now turn our attention to a special case where it is possible to construct a multivariate Markov process that contains $Z$.

\subsubsection{Construction of a multivariate Markov process}\label{sec:multivariate_markov}
We now consider the case where each component of $a$ is either constant or linear in $t$. As an example, we take the model from Section \ref{sec:stochastic_termination} with basis functions 
\begin{align*}
 &\phi^1(x) =  \frac{64-x}{39} \quad  \text{, } \quad \phi^2(x) = \frac{x-25}{39}\\
 &\psi^1(t) = 1 \quad  \text{, } \quad \psi^2(t) = t.
\end{align*}
\noindent
Aro, Djehiche and L\"ofdahl \cite{ADL13} fit this model to data from a Swedish insurance company and suggest that it can be seen as a middle ground model when considering goodness of fit versus tractability. Here, it proves to be an interesting special case which allows us to construct a multivariate Markov process from a non-Markovian one. 

Consider the vector valued Markov process $Z=(Z^1,Z^2)$ defined by
\begin{align}\label{eq:Z_system}
\begin{cases}
Z^1_t &= a^T\nu_t,\\
Z^2_t &= \dot a^T\nu_t.
\end{cases}
\end{align}
\noindent
The process $Z$ satisfies the system of stochastic differential equations
\begin{align}\label{eq:dZ_system}
\begin{cases}
dZ^1_t &= (Z^2_t + a^T\mu)dt + a^TAdW_t,\\
dZ^2_t &= \dot a^T\mu dt + \dot a^TAdW_t.
\end{cases}
\end{align}
\noindent
By \cite[Theorem 8.4.3]{oksendal2003stochastic}, $Z$ is identical in law to the process $\widehat Z=(\widehat Z^1,\widehat Z^2)$ defined by
\begin{align}\label{eq:dZhat}
d\widehat Z_t = \alpha(t,\widehat Z_t)dt + \gamma(t)d\widehat W_t,
\end{align}
\noindent
where 
\begin{align}\label{eq:oksendal_cond2dim}
\alpha(t,Z_t) &= E[\left(
\begin{array}{c}
Z^2_t + a^T\mu\\
\dot a^T\mu
\end{array} \right)
| \F^Z_t] = \left(
\begin{array}{c}
Z^2_t + a^T\mu\\
\dot a^T\mu
\end{array} \right),\\
\gamma(t)\gamma(t)^T &= \left(
\begin{array}{c}
a^TA\\
\dot a^TA
\end{array} \right)
\left(
\begin{array}{c}
a^TA\\
\dot a^TA
\end{array} \right)^T = \left(
\begin{array}{cc}
a^T\Sigma a & a^T\Sigma \dot a \\
a^T\Sigma \dot a & \dot a^T\Sigma \dot a
\end{array} \right),
\end{align}
\noindent
and $\widehat W$ is a two-dimensional standard Wiener process. Thus, we have effectively reduced the process $\nu$ to the two-dimensional process $\widehat Z$, and we may compute moments of present values by solving the PDEs \eqref{eq:moments_pde} with the generator $\widehat \A$ of $\widehat Z$ and termination intensity $q(t,\widehat Z_t) = f(\widehat Z^1_t)$. 

This recipe can easily be extended to the case where $a^{(k)}$, the $k$'th derivative of $a$ w.r.t. time, is constant. Then, the system \eqref{eq:dZ_system} becomes a system of $k+1$ SDEs, and the process defined by \eqref{eq:dZhat} will have $k+1$ driving Wiener processes. Thus, if $k+1<nm$, that is, if the number of driving Wiener processes is smaller than the number of parameters in the statistical model, the dimensionality of the problem can be reduced, while still preserving all probabilistic properties of the system.

\subsubsection{Mimicking the killed environment process}
It is not always possible to construct a multivariate Markov process containing $Z$ as above, and even if it is possible, it is not certain that the number of dimensions will be reduced. For example consider the model from Section \ref{sec:stochastic_termination} with basis functions 
\begin{align*}
 &\phi^1(x) =  \frac{64-x}{39} \quad  \text{, } \quad \phi^2(x) = \frac{x-25}{39}\\
 &\psi^1(t) = 1 \quad  \text{, } \quad \psi^2(t) = e^{-t} \quad  \text{, } \quad \psi^3(t) = e^{-2t}.
\end{align*}
\noindent
It is immediate that we cannot apply the recipe of Section \ref{sec:multivariate_markov}. As an alternative, we will rely on an idea suggested by Krylov \cite{Krylov} to construct a Markov process $\widehat Z$ that mimics certain features of the behavior of the process $Z$ such as
\begin{equation}
\widehat Z_t \eqdist Z_t, \hspace{5pt} t\geq0.
\end{equation} 
\noindent
Proposition \ref{Krylov} below displays a general result about existence of the Markov process $\widehat Z$ when $Z$ is a non-Markov diffusion. This result appeared first in Krylov \cite{Krylov} and extended in Gy\"ongy \cite{Gyongy} and Borkar \cite{Borkar89} and generalized in various ways to L{\'e}vy processes and semimartingales in Bhatt and Borkar \cite{BhattBorkar}, Kurtz and Stockbrigde \cite{KurtzStockbridge98,KurtzStockbridge01}, Bentata and Cont \cite{BentataCont}, and Bouhadou and Ouknine \cite{BouhadouOuknine}. The process $\widehat Z$ is often called {\it Markovian projection} or {\it mimicking process} of $Z$.

\begin{proposition} (Kurtz and Stockbrigde \cite{KurtzStockbridge98}, Corollary 4.3) \label{Krylov}\\
\noindent
When $Z$ satisfies
\begin{equation}\label{sde}
Z_t=Z_0+\int_0^ t\beta(s)ds+\int_0^ t\delta(s)dW_s,
\end{equation}
where, W is an $\mathbb{R}^d$-valued ${\cal F}_t$-Brownian motion; $\delta$ and $\beta$ are measurable, ${\cal F}_t$-adapted processes taking values in the set of $d\times d$ matrices $\mathbb{M}^{d\times d}$ and $\mathbb{R}^d$, respectively;
and $Z_0$ is $\mathbb{R}^d$-valued and ${\cal F}_0$-measurable. Then there exist measurable functions $\sigma: [0,\infty)\times\mathbb{R}^d \mapsto \mathbb{M}^{d\times d}$ and $b:\mathbb{R}^d\mapsto \mathbb{R}^d$, an
$\mathbb{R}^d$-valued Brownian motion $\widehat W$, and a process $\widehat Z$ satisfying
\begin{equation}\label{sde-m-1}
\widehat Z_t=\widehat Z_0+\int_0^ tb(s,\widehat Z_s)ds+\int_0^ t\sigma(s,\widehat Z_s)d\widehat W_s,
\end{equation}
such that  for each $t\ge 0$,
\begin{equation}\label{mimicking-2}
(Z_t,E[\beta(t)|Z_t],E[\delta(t)\delta^T(t)|Z_t])\stackrel{d}=(\widehat Z_t, b(t,\widehat Z_t), \sigma(t,\widehat Z_t)\sigma^T(t,\widehat Z_t)).
\end{equation} 
\end{proposition}
\noindent
For the sequel, we set
\begin{equation}\label{eq:t-reserve}
v^Z(t,z)=v(t,z),
\end{equation}
whenever $Z_t=z$. When the intensity $q$ is a {\it constant}, it is immediate that 
\begin{equation}\label{mimicking-1}
v^{\widehat Z}(t,\widehat Z_t)\eqdist v^Z(t,Z_t)
\end{equation}
\noindent
holds whenever the property (\ref{mimicking-2}) remains true. A counter-example constructed by Borkar \cite{Borkar89} suggests that it is not always possible to obtain a Markov process $\widehat Z$ whose finite dimensional distributions agree with those of the process $Z$. Therefore, (\ref{mimicking-1}) may not hold when the discount factor $q$ depends of $Z$. Kurtz and Stockbrigde \cite[Theorem 5.1]{KurtzStockbridge98} do construct a Markov process $\widehat Z$ for which (\ref{mimicking-1}) holds, even when the discount factor $q$ depends on $Z$, but the $t$-marginal distributions of $\widehat Z$ and $Z$ may not be identical i.e. $\widehat Z$ does not mimic $Z$.

A closer look at  the $t$-reserve $v^Z(t,Z_t)$ suggests that we should mimic the process $\bar Z$ obtained by 'killing' $Z$ at rate $q$ in the sense described e.g. in Rogers and Williams \cite[Section III.18]{RogersWilliams}. The intuitive idea behind killing is that $\bar Z$ agrees with $Z$ up to time $\bar \tau$ and $\bar Z_t=\partial, \ t\geq \bar\tau$, where $\partial$ is some absorbing state, and
\begin{equation}\label{killed-0}
\bar P(\bar\tau>t|\F^Z_t) = e^{-\int_0^t q(s,Z_s)ds}.
\end{equation}
\noindent
Given a process $Z$ on $(\Omega, {\cal F}, {\cal F}_t, P)$, the process $\bar Z$ obtained by 'killing' $Z$ at rate $q$ is defined on a probability space $(\bar\Omega, \bar{\cal F}, \bar{\cal F}_t, \bar P)$ by 
\begin{equation}\label{killed-1}
\bar P(\bar Z_t\in A) := E[M_t I_{\{Z_t\in A\}}], 
\end{equation}
where $M_t:=e^{-\int_0^t q(s,Z_s)ds}$. Moreover, for any Borel measurable and bounded function $f$,
\begin{equation}\label{killed-2}
\begin{array}{lll}
\bar E[f(\bar Z_t)|\bar Z_s=z]&=E[f(Z_t)\frac{M_t}{M_s}| Z_s=z]\\ &=E[f(Z_t)e^{-\int_s^t q(u,Z_u) du}| Z_s=z].
\end{array}
\end{equation}

\noindent If $Z$ is given by (\ref{sde}), letting
$$
{\cal L}_tf(x):=\beta(t)\nabla f(x)+\frac{1}{2}tr\left(\delta(t)\delta^T(t) \nabla^2 f(x)\right)-q(t,x)f(x),
\quad f\in{\cal C}_0^{\infty}(\mathbb{R}^d),
$$
applying It\^o's formula to $M_tf(Z_t)$ and taking expectation, we get
$$
E[M_tf(Z_t)]=E[f(Z_0)]+\int_0^t E[M_s{\cal L}_sf(Z_s)]ds.
$$
Thus, in view of (\ref{killed-1}), we have
\begin{equation}\label{forward}
\begin{array}{lll}
\bar E[f(\bar Z_t)]&=\bar E[f(\bar Z_0)]+\int_0 \bar E[{\cal L}_sf(\bar Z_s)]ds\\
&=\bar E[f(\bar Z_0)]+\int_0 \bar E[\bar E[{\cal L}_sf(\bar Z_s)|\bar Z_s]]ds \\ & =
\bar E[f(\bar Z_0)]+\int_0 \bar E[\widehat {\cal A}_sf(\bar Z_s)]ds,
\end{array}
\end{equation}
where, 
\begin{equation}\label{MG}
\widehat {\cal A}_tf(x):={\cal A}_tf(x)-q(t,x)f(x)
\end{equation}
and 
\begin{equation}\label{G}
{\cal A}_tf(x)=:\bar b(t,x)\nabla f(x)+\frac{1}{2}tr\left(\bar\sigma\bar\sigma^T(t,x)\nabla^2 f(x)\right),
\end{equation}
with
\begin{equation}\label{mimicking-3}
\begin{array}{lll}
\bar b(t,x):=\bar E[\beta(t)|\bar Z_t=x]=E[M_t\beta(t)|Z_t=x],\\  \\
\bar \sigma\bar \sigma^T(t,x):=\bar E[\delta(t)\delta^T(t)|\bar Z_t=x]=E[M_t\delta(t)\delta^T(t)|Z_t=x].
\end{array}
\end{equation} 
In view of Proposition \ref{Krylov}, then there exist an $\mathbb{R}^d$-valued Brownian motion $B$, and a process $\widehat {\bar Z}$ satisfying
\begin{equation}\label{sde-m-2}
\widehat {\bar Z}_t=\widehat {\bar Z}_0+\int_0^ tb(s,\widehat{\bar Z}_s)ds+\int_0^ t\sigma(s,\widehat{\bar Z}_s)dB_s,
\end{equation}
whose infinitesimal generator is $\widehat {\cal A}$, such that,  for each $t\ge 0$,
\begin{equation}\label{mimicking-4}
(\bar Z_t,\bar E[\beta(t)|\bar Z_t],\bar E[\delta(t)\delta^T(t)|\bar Z_t])\stackrel{d}=(\widehat{\bar Z}_t, b(t,\widehat{\bar Z}_t), \sigma(t,\widehat{\bar Z}_t)\sigma^T(t,\widehat{\bar Z}_t)).
\end{equation} 
\noindent 
In terms of the mimicked killed Markov diffusion process $\widehat{\bar Z}$, using (\ref{killed-2}) and (\ref{mimicking-4}), we have the following property for the $t$-reserve:
\begin{equation}\label{eq:mimicking-5}
\begin{array}{lll}
E[v^Z(t,Z_t)] & =E[\int_t^Te^{-\int_t^sr(u)du} E[e^{-\int_t^s q(u,Z_u)du}g(s,Z_s)|Z_t]ds] \\ 
 &=\int_t^Te^{-\int_t^sr(u)du} \bar E[g(s,\bar Z_s)]ds \\
 &=\int_t^Te^{-\int_t^sr(u)du} \bar E[g(s,\widehat{\bar Z}_s)]ds\\ 
& =\int_t^Te^{-\int_t^sr(u)du} E[e^{-\int_t^s q(u,\widehat{ Z}_u)du}g(s,\widehat{ Z}_s)]ds \\
&=E[v^{\widehat Z}(t,\widehat{ Z}_t)].
\end{array}
\end{equation} 
Applying the Feynman-Kac formula, $v=v^{\widehat Z}$ satisfies the following PDE
\begin{equation}\left\{\begin{array}{lll}
\frac{\partial v}{\partial s}(s,x)+\widehat{\cal A}_s v(s,x)+g(s,x)=r(s)v(s,x),\,\, t\le s<T,\\ \\ v(T,x)=0.
\end{array}\right.
\end{equation}
\noindent
Hence, using (\ref{MG}), we get 
\begin{equation}\label{eq:PDE-m1}
\left\{\begin{array}{lll}
\frac{\partial v}{\partial s}(s,x)+{\cal A}_s v(s,x)+g(s,x)=(q(s,x)+r(s))v(s,x),\,\, t\le s<T,\\ \\ v(T,x)=0.
\end{array}
\right.
\end{equation}
\noindent
Note that \eqref{eq:mimicking-5} does not imply that $v^Z(t,x)=v^{\widehat Z}(t,x)$ for all $x$, only that, 'on average over all $x$' they will agree. A way to think of this is that if $v^{ Z}(t, Z_t)$ is an unbiased estimator of some parameter $\theta$, then $v^{\widehat Z}(t,\widehat{ Z}_t)$ is also an unbiased estimator of $\theta$. For all purposes, the PDE (\ref{eq:PDE-m1}) is only useful if we can explicitly compute the terms $\bar b$ and $\bar \sigma$ displayed in (\ref{mimicking-3}), which is in general out of reach even for the simplest Gaussian dynamics, due to presence of the path-dependent discounting factor $M$. This makes the idea of mimicking the killed process less attractive. We make one final attempt in constructing a mimicking process that preserves some properties of $V_t$.

\subsubsection{Mimicking the environment process}\label{sec:mimicking}
We suggest the following recipe for computing an approximate $t$-reserve. First, we determine the Markovian projection $\widehat Z$ of the underlying process $Z$. Then, we consider the moments of $\widehat V_t$ defined by 
\begin{equation}
v_n(t,z)=E^{t,z}[\widehat V^n_t],
\end{equation}
which satisfies \eqref{eq:moments_pde}, as an approximation of the true moments based on $Z$. Using Proposition \ref{Krylov}, letting
\begin{align}\label{eq:alpha_and_gamma}
\alpha(t,z) &:= E[\beta(t)|Z_t = z] = E[\dot a(t)^T\nu_t + a(t)^T\mu|Z_t = z]\nonumber\\
\gamma(t) &:=  \sqrt{E[\delta(t)\delta(t)^T| Z_t = z]} = \sqrt{\delta(t)\delta(t)^T} = \sqrt{a(t)^TAA^Ta(t)},
\end{align}
\noindent
then the process $\widehat Z$ defined by
\begin{equation}\label{hatZ_dynamics}
d\widehat Z_t = \alpha(t,\widehat Z_t)dt + \gamma(t)d\widehat W_t,
\end{equation}
where $\widehat W$ is a standard Brownian motion, has the same marginal distributions as the process $Z$. However, this does not imply that $V_t$ and $\widehat V_t$ have the same marginals. In the numerical results section below, we study the distributions of $V_t$ and $\widehat V_t$ by Monte Carlo simulation of the processes $\nu$ and $\widehat Z$, respectively. It turns out that the distributions are almost identical, and we proceed with this mimicking approach. It then remains to determine the function $\alpha$. We have
\begin{align}\label{eq:findalpha_1}
\alpha(t,z) &= E[\beta(t)|Z_t = z] = E[\dot a^T\nu_t + a^T\mu|a^T\nu_t = z]\nonumber\\
&= a^T\mu + \dot a^T(\nu_0+\mu t)\nonumber\\
&+ E[\dot a^TAW_t|a^TAW_t = z-a^T(\nu_0+\mu t)].
\end{align}
\noindent
Since all linear combinations of $W_t$ are Gaussian, we have
\begin{align}\label{eq:findalpha_2}
E[\dot a^TAW_t|a^TAW_t = z-a^T(\nu_0+\mu t)]=(z-a^T(\nu_0+\mu t))\frac{\text{Cov}(\dot a^TAW_t,a^TAW_t)}{\text{Var}(a^TAW_t)}.
\end{align}
\noindent
Using the independence of the marginal distributions of the components of $W$, we have
\begin{align}\label{eq:findalpha_3}
\text{Var}(a^TAW_t) &= \text{Var}(\sum_jW_t^j\sum_ia_iA_{ij})\nonumber\\
&= \sum_j\text{Var}(W_t^j)(\sum_ia_iA_{ij})^2 = ta^TAA^Ta = ta^T\Sigma a.
\end{align}
\noindent
Similarly,
\begin{align}\label{eq:findalpha_4}
\text{Cov}(\dot a^TAW_t,a^TAW_t) = ta^T\Sigma \dot a.
\end{align}
\noindent
Finally, we obtain the following explicit expression for $\alpha$,
\begin{align}\label{eq:findalpha_5}
\alpha(t,z) = a^T\mu + \dot a^T(\nu_0+\mu t) + (z-a^T(\nu_0+\mu t))\frac{a^T\Sigma \dot a}{a^T\Sigma a}.
\end{align}
\noindent
Curiously, it happens that $\widehat Z$ is a Hull-White process, a model form which allows for explicit pricing of discount factors, see Hull and White \cite{HullWhite}. Here, the hazard rate is given by the non-negative process $f(\widehat Z_t)$, which is no longer of Hull-White form. Hence, we are unable to exploit the tractability of the Hull-White model. This is not necessarily a bad thing, since the Hull-White process allows for negative hazard rates, a property which is not always desired. Still, we may use Proposition \ref{prop:moments_pde} to compute moments of $\widehat V_t$. From the representation \eqref{hatZ_dynamics}, \eqref{eq:moments_pde} becomes
\begin{align}\label{eq:moments_pde_adl}
\begin{cases}
-\frac{\partial v_n}{\partial s} + n(f(z)+r(s))v_n = \alpha(s,z)\frac{ \partial v_n}{\partial z}+\half\gamma^2(s)\frac{\partial^2v_n}{\partial z^2} + ng(s,z)v_{n-1}, \hspace{10pt} t\leq s<T\\
v_n(T,z) = 0,
\end{cases}
\end{align}
\noindent
with $f$, $\alpha$ and $\gamma$ given by \eqref{eq:f}, \eqref{eq:findalpha_5} and \eqref{eq:alpha_and_gamma}, respectively. The PDE \eqref{eq:moments_pde_adl} can be solved using numerical methods, e.g. finite-difference schemes.

\section{Numerical results}\label{sec:numerics}
In this section, we implement two disability termination models together with the dimension reduction techniques of Sections \ref{sec:multivariate_markov} and \ref{sec:mimicking}. The parameters of the models for the years 2000-2011 are estimated using the method from \cite{ADL13}. Using Monte Carlo simulations, the distribution of the functional $V_t$ is compared to the distributions of $\widehat V_t^1$ and $\widehat V_t^2$, where $\widehat V_t^1$ denotes the functional of the multivariate Markov process constructed in Section \ref{sec:multivariate_markov}, and $\widehat V_t^2$ denotes the functional of the Markov projection process of Section \ref{sec:mimicking}. Further, the PDE \eqref{eq:moments_pde_adl} is used to compute the first three moments of $\widehat V_t^2$, where we have chosen the parameters $x=55$, $T=10$, $r=0.02$, $t=0$, $g(t,z) = 1$. The PDE is solved using a first order implicit finite-difference scheme, and the results are compared to a Monte Carlo simulation with 100,000 draws and $\Delta t = 0.01$.

\subsection{A linear model}
We consider the model from Section \ref{sec:example} with basis functions 
\begin{align*}
 &\phi^1(x) =  \frac{64-x}{39} \quad  \text{, } \quad \phi^2(x) = \frac{x-25}{39}\\
 &\psi^1(t) = 1 \quad  \text{, } \quad \psi^2(t) = t.
\end{align*}
\noindent
We assume that $\nu$ follows a 4-dimensional Brownian motion, and estimate the drift and covariance matrix from the time series of parameter values. 

The densities and distribution functions of $V_t$, $\widehat V_t^1$ and $\widehat V_t^2$ are presented in Figures \ref{fig:pdf_v_vs_multivariate_markov}-\ref{fig:cdf_v_vs_mimick}. Note that, due to confidentiality, the $x$-axes are presented as fractions of the Best Estimate anno 2011. Here, the Best Estimate is defined as the value of the initial reserve assuming that the model parameters are held constant over the entire policy period. 

As can be seen in the plots, the density- and distribution functions of $\widehat V_t^1$ and $\widehat V_t^2$ are almost identical to those of $V_t$. Indeed, using a standard two-sample Kolmogorov-Smirnov test, we cannot reject the hypothesis that the samples of $V_t$ and $\widehat V_t^1$ are drawn from the same distribution. The corresponding $p$-value is 0.56. However, we can in fact reject the hypothesis that the samples of $V_t$ and $\widehat V_t^2$ are drawn from the same distribution. Still, we conclude that we can consider $\widehat V_t^2$ as an approximation of $V_t$, and that, as expected, $V_t$ and $\widehat V_t^1$ have identical distributions. This is a highly useful result since it reduces the dimensionality of the problem, which significantly reduces the computational cost. In this example, the choice stands between obtaining an exact result with two space dimensions, or an approximate result with one space dimension, compared to the four space dimensions of the original problem.

Numerical values from the PDE solver for the first three moments, as a fraction of the Best Estimate anno 2011, are presented in Table \ref{table:moments_pde}. The values of $v_1$ correspond to the initial reserve. We present the values as fractions of the Best Estimate rather than monetary units due to confidentiality. 99$\%$ approximate confidence intervals of moments of $V_t$, $\widehat V_t^1$ and $\widehat V_t^2$ from the Monte Carlo simulation are presented in Table \ref{table:moments_mc1}. As we can see, the moments from the PDE solver lie well within the 99$\%$ confidence intervals from the Monte Carlo simulation of $\widehat V_t^2$, and a few percentage points above the 99$\%$ confidence intervals from the Monte Carlo simulation of $V_t$. We stress the fact that we are trading accuracy for computational efficiency.

\begin{table}[!ht]
\caption{Moments of $\widehat V_t^2$ from the PDE solver, scaled by the Best Estimate.}
\label{table:moments_pde}
\footnotesize
\begin{center}
\begin{tabular}{c|ccc}
 $\Delta z = \Delta t$ & $v_1$ & $v_2$ & $v_3$ \\
\hline
0.1 & 0.9097 & 0.8019 & 0.7567\\
0.05 & 0.9064 & 0.7929 & 0.7370\\
0.01 & 0.9040 & 0.7865 & 0.7239\\
0.005 & 0.9037 & 0.7858 & 0.7226\\
0.001 & 0.9035 & 0.7853 & 0.7217\\
\end{tabular} \\
\end{center}
\end{table}

\begin{table}[!ht]
\caption{99$\%$ approximate confidence intervals of moments of $V_t$, $\widehat V_t^1$ and $\widehat V_t^2$ from the Monte Carlo simulation.}
\label{table:moments_mc1}
\footnotesize
\begin{center}
\begin{tabular}{c|ccc}
 & $v_1$ & $v_2$ & $v_3$ \\
\hline
MC, $V_t$            & (0.8986    0.9011) & (0.7720    0.7772) & (0.6938    0.7030)\\
MC, $\widehat V_t^1$ & (0.8991    0.9016) & (0.7726    0.7778) & (0.6944    0.7035)\\
MC, $\widehat V_t^2$ & (0.9013    0.9041) & (0.7812    0.7870) & (0.7148    0.7257)\\
\end{tabular} \\
\end{center}
\end{table}

\begin{figure}[!ht]
\begin{center}
\epsfig{file=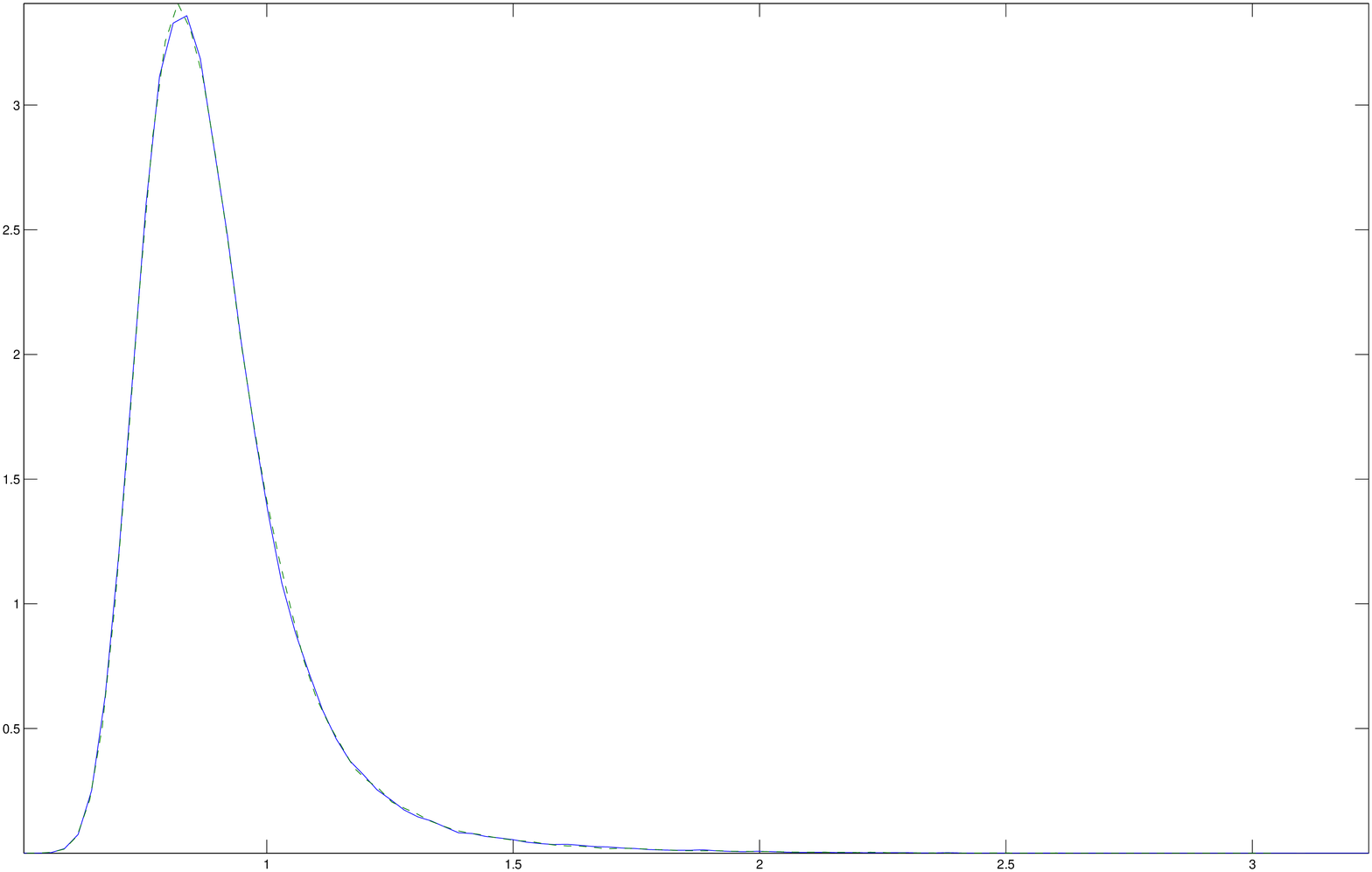,height=0.7\linewidth, width=\linewidth, angle=0}
\vspace{-25pt}
\caption{The densities of $V_t$ (solid) and $\widehat V_t^1$ (dashed).}
\label{fig:pdf_v_vs_multivariate_markov}
\end{center}
\end{figure}

\begin{figure}[!ht]
\begin{center}
\epsfig{file=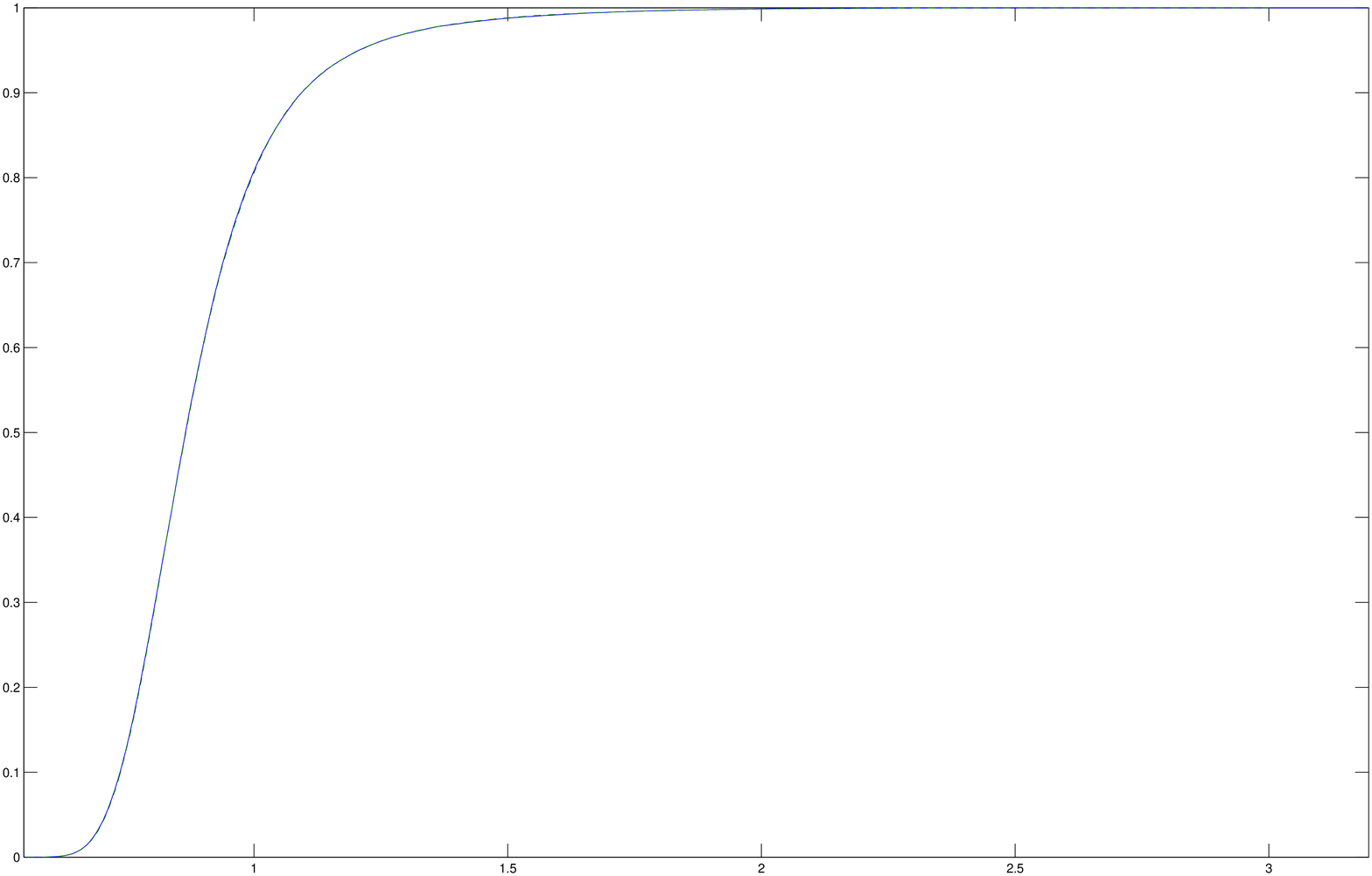,height=0.7\linewidth, width=\linewidth, angle=0}
\vspace{-25pt}
\caption{The distribution functions of $V_t$ (solid) and $\widehat V_t^1$ (dashed).}
\label{fig:cdf_v_vs_multivariate_markov}
\end{center}
\end{figure}

\begin{figure}[!ht]
\begin{center}
\epsfig{file=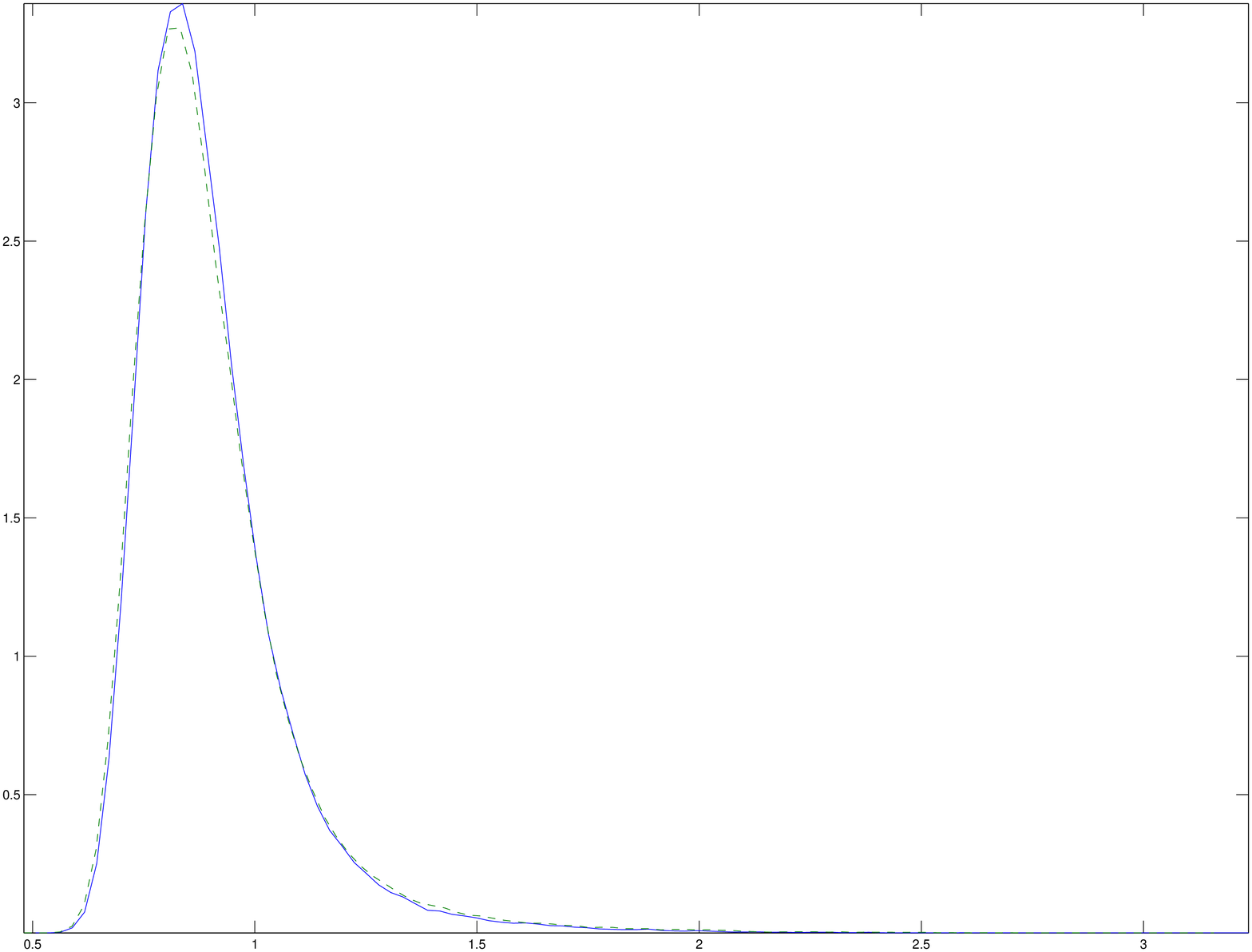,height=0.7\linewidth, width=\linewidth, angle=0}
\vspace{-25pt}
\caption{The densities of $V_t$ (solid) and $\widehat V_t^2$ (dashed).}
\label{fig:pdf_v_vs_mimick}
\end{center}
\end{figure}

\begin{figure}[!ht]
\begin{center}
\epsfig{file=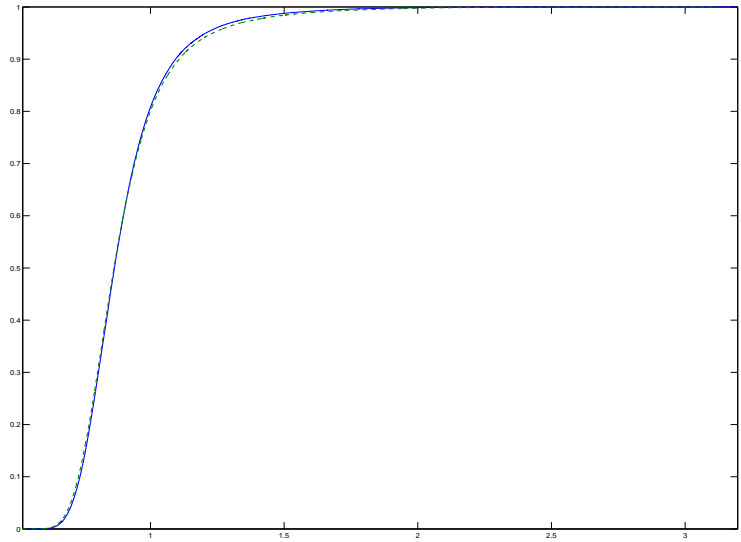,height=0.7\linewidth, width=\linewidth, angle=0}
\vspace{-25pt}
\caption{The distribution functions of $V_t$ (solid) and $\widehat V_t^2$ (dashed).}
\label{fig:cdf_v_vs_mimick}
\end{center}
\end{figure}

\clearpage
\subsection{A non-linear model}
Next, we consider the model from Section \ref{sec:example} with basis functions 
\begin{align*}
 &\phi^1(x) =  \frac{64-x}{39} \quad  \text{, } \quad \phi^2(x) = \frac{x-25}{39}\\
 &\psi^1(t) = 1 \quad  \text{, } \quad \psi^2(t) = e^{-t} \quad  \text{, } \quad \psi^3(t) = e^{-2t}.
\end{align*}
\noindent
Using the method from \cite{ADL13}, this model yields slightly better goodness of fit compared to the linear model. We assume that $\nu$ follows a 6-dimensional Brownian motion, and estimate the drift and covariance matrix from the time series of parameter values.

For this non-linear model, it is immediate that we cannot implement the recipe of Section \ref{sec:multivariate_markov} to reduce the dimensionality. Instead, we focus our efforts on the Markov projection technique of Section \ref{sec:mimicking}. 

The densities and distribution functions of $V_t$ and $\widehat V_t^2$ are presented in Figures \ref{fig:pdf_v_vs_mimick_6par}-\ref{fig:cdf_v_vs_mimick_6par}, and 99$\%$ approximate confidence intervals of moments of $V_t$ and $\widehat V_t^2$ from the Monte Carlo simulation are presented in Table \ref{table:moments_mc2}. It is apparent that the mimicking approximation performs slightly worse for this model compared to the linear model. However, comparing Table \ref{table:moments_mc2} with Table \ref{table:moments_mc1}, it seems that the approximation error and the model uncertainty are of the same magnitude: the deviations between the linear model and the non-linear model are similar to the deviations between any one of the models and its corresponding Markovian projection, at least for the first two moments. For the non-linear model, the Markovian projection shows a significant approximation error for the third moment. Again, we stress the fact that we are trading accuracy for computational efficiency. As the Markov projection technique seems to slightly overestimate both the moments and the thickness of the tail of $V_t$, it could possibly be used to obtain conservative risk estimates, although further research is needed to confirm this hypothesis.

\begin{table}[!ht]
\caption{99$\%$ approximate confidence intervals of moments of $V_t$ and $\widehat V_t^2$ from the Monte Carlo simulation.}
\label{table:moments_mc2}
\footnotesize
\begin{center}
\begin{tabular}{c|ccc}
 & $v_1$ & $v_2$ & $v_3$ \\
\hline
$V_t$ & (0.9292    0.9313) & (0.8156    0.8200) & (0.7368    0.7441)\\
$\widehat V_t^2$ & (0.9381    0.9415) & (0.8605    0.8688) & (0.8583    0.8768)\\
\end{tabular} \\
\end{center}
\end{table}

\begin{figure}[!ht]
\begin{center}
\epsfig{file=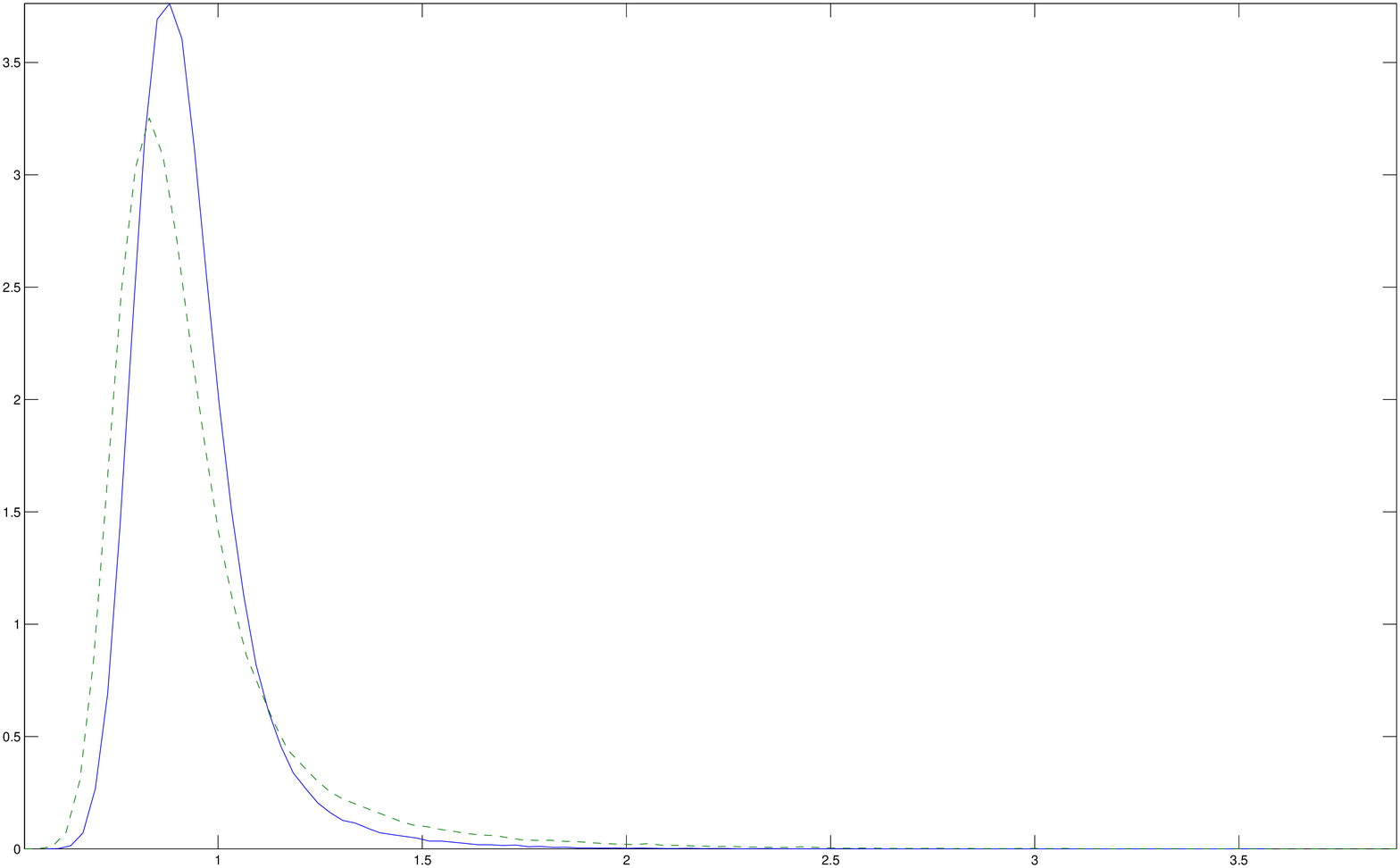,height=0.7\linewidth, width=\linewidth, angle=0}
\vspace{-25pt}
\caption{The densities of $V_t$ (solid) and $\widehat V_t^2$ (dashed).}
\label{fig:pdf_v_vs_mimick_6par}
\end{center}
\end{figure}

\begin{figure}[!ht]
\begin{center}
\epsfig{file=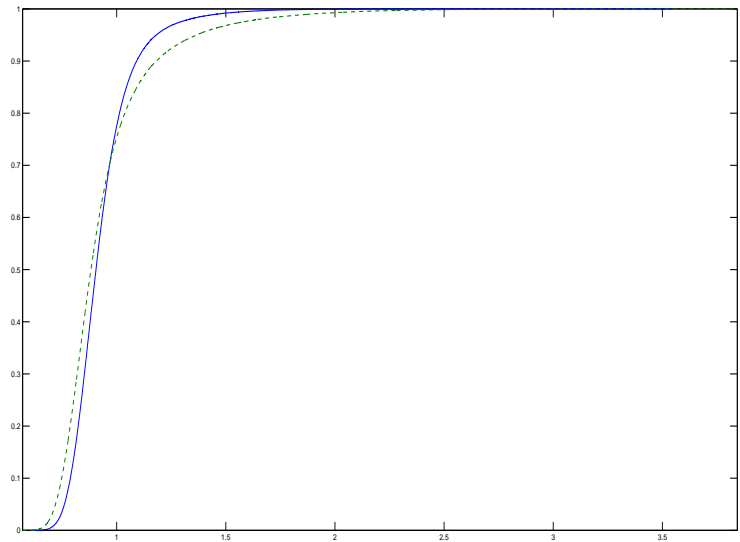,height=0.7\linewidth, width=\linewidth, angle=0}
\vspace{-25pt}
\caption{The distribution functions of $V_t$ (solid) and $\widehat V_t^2$ (dashed).}
\label{fig:cdf_v_vs_mimick_6par}
\end{center}
\end{figure}

\clearpage
\section{Acknowledgements}
The first author gratefully acknowledges financial support from the Swedish Export Credit Corp. (SEK). The second author gratefully acknowledges financial support from the Filip Lundberg  and Eir's 50 Years foundations. Both authors appreciate the helpful comments of an anonymous referee.

\clearpage
\newpage
\bibliography{disability_reserve}

\begin{thebibliography}{10}

\bibitem{ADL13}
{\sc Aro, H., Djehiche, B., and L\"ofdahl, B.}
\newblock Stochastic modelling of disability insurance in a multi-period
  framework.
\newblock {\em Scandinavian Actuarial Journal\/} (2013).

\bibitem{BentataCont}
{\sc Bentata, A., and Cont, R.}
\newblock {Mimicking the marginal distributions of a semimartingale}.
\newblock {\em arXiv:0910.3992\/} (2012).

\bibitem{BhattBorkar}
{\sc Bhatt, A., and Borkar, V.}
\newblock Occupation measures for controlled markov processes: Characterization
  and optimality.
\newblock {\em Annals of Probability 24\/} (1996), 1531--1562.

\bibitem{Biffis05}
{\sc Biffis, E.}
\newblock Affine processes for dynamic mortality and actuarial valuations.
\newblock {\em Insurance: Mathematics and Economics 37\/} (2005), 443--468.

\bibitem{Borkar89}
{\sc Borkar, V.}
\newblock Mimicking finite dimensional marginals of a controlled diffusion by
  simpler controls.
\newblock {\em Stochastic Processes and their Applications 31}, 2 (1989),
  333--342.

\bibitem{BouhadouOuknine}
{\sc Bouhadou, S., and Ouknine, Y.}
\newblock Mimicking finite dimensional marginals of a controlled diffusion with
  jumps.
\newblock {\em Stochastics and Dynamics 14}, 1 (2014).

\bibitem{Christiansen}
{\sc Christiansen, M., Denuit, M., and Lazar, D.}
\newblock The {S}olvency {II} square-root formula for systematic biometric
  risk.
\newblock {\em Insurance: Mathematics and Economics 50\/} (2012), 257--265.

\bibitem{Dahl04}
{\sc Dahl, M.}
\newblock Stochastic mortality in life insurance: market reserves and
  mortality-linked insurance contracts.
\newblock {\em Insurance: Mathematics and Economics 35\/} (2004), 113--136.

\bibitem{DahlMoller}
{\sc Dahl, M., and M{\o}ller, T.}
\newblock Valuation and hedging of life insurance liabilities with systematic
  mortality risk.
\newblock {\em Insurance: Mathematics and Economics 39\/} (2006), 193--217.

\bibitem{friedman2012stochastic}
{\sc Friedman, A.}
\newblock {\em Stochastic differential equations and applications}.
\newblock Courier Dover Publications, 2012.

\bibitem{Gyongy}
{\sc Gy\"ongy, I.}
\newblock Mimicking the one-dimensional marginal distributions of processes
  having an {It\^o} differential.
\newblock {\em Probability Theory and Related Fields 71}, 4 (1986), 501--516.

\bibitem{HullWhite}
{\sc Hull, J., and White, A.}
\newblock {Pricing Interest-Rate-Derivative Securities}.
\newblock {\em The Review of Financial Studies 3}, 4 (1990), 573--592.

\bibitem{HyndmanUllah}
{\sc Hyndman, R., and Ullah, M.~S.}
\newblock Robust forecasting of mortality and fertility rates: {A} functional
  data approach.
\newblock {\em Computational Statistics and Data Analysis 51\/} (2007),
  4942--4956.

\bibitem{Krylov}
{\sc Krylov, N.}
\newblock {Once more about the connection between elliptic operators and
  It\^o's stochastic equations}.
\newblock In {\em {Statistics and control of stochastic processes, Steklov
  Seminar}\/} (1984), pp.~214--229.

\bibitem{KurtzStockbridge98}
{\sc Kurtz, T.~G., and Stockbridge, R.~H.}
\newblock Existence of markov controls and characterization of optimal markov
  controls.
\newblock {\em {SIAM} Journal on Control and Optimization 36\/} (1998),
  609--653.

\bibitem{KurtzStockbridge01}
{\sc Kurtz, T.~G., and Stockbridge, R.~H.}
\newblock Stationary solutions and forward equations for controlled and
  singular martingale problems.
\newblock {\em Electronic Journal of Probability 6}, 17 (2001), 1--52.

\bibitem{LevantesiMenzietti}
{\sc Levantesi, S., and Menzietti, M.}
\newblock Managing longevity and disability risks in life annuities with long
  term care.
\newblock {\em Insurance: Mathematics and Economics 50\/} (2012), 391--401.

\bibitem{LudkovskiYoung}
{\sc Ludkovski, M., and Young, V.~R.}
\newblock Indifference pricing of pure endowments and life annuities under
  stochastic hazard and interest rates.
\newblock {\em Insurance: Mathematics and Economics 42\/} (2008), 14--30.

\bibitem{Norberg95}
{\sc Norberg, R.}
\newblock Differential equations for moments of present values in life
  insurance.
\newblock {\em Insurance: Mathematics and Economics 17\/} (1995), 171--180.

\bibitem{Norberg}
{\sc Norberg, R.}
\newblock Forward mortality and other vital rates - {A}re they the way forward?
\newblock {\em Insurance: Mathematics and Economics 47\/} (2010), 105--112.

\bibitem{oksendal2003stochastic}
{\sc {\O}ksendal, B.}
\newblock {\em Stochastic differential equations}.
\newblock Springer, 2003.

\bibitem{PrakasaRao}
{\sc Prakasa~Rao, B. L.~S.}
\newblock Conditional independence, conditional mixing and conditional
  association.
\newblock {\em Annals of the Institute of Statistical Mathematics 61}, 2
  (2009), 441--460.

\bibitem{RogersWilliams}
{\sc Rogers, L., and Williams, D.}
\newblock {\em {Diffusions, Markov processes and martingales. Volume one:
  Foundations}}.
\newblock Wiley, 1995.

\bibitem{VasicekRisk}
{\sc Vasicek, O.}
\newblock The distribution of loan portfolio value.
\newblock {\em Risk 15}, 12 (2002), 160--162.

\end{thebibliography}
\bibliographystyle{acm}

\end{document}